%% file: sofsem2027.tex
\documentclass[runningheads]{llncs}
\usepackage[T1]{fontenc}
\usepackage{amsmath,amssymb,mathtools}
\usepackage{algorithm}
\usepackage{algpseudocode}
\usepackage{todonotes}
\usepackage{tikz}
\usetikzlibrary{decorations.pathreplacing}
\usepackage{hyperref}
\input{macro_Tikz}

\newcommand{\E}{\mathbb{E}}
\newcommand{\Prb}{\mathbb{P}}
\newcommand{\HT}{\mathrm{HT}}

\title{Optimizing Both Checking and Update Costs in Random Walk Search\thanks{Appears in the Proceedings of the 52nd International Conference on Current Trends in Theory and Practice of Computer Science (Sofsem 2027).}}
\titlerunning{Optimizing Checking and Update Costs}

\author{Simon Apers$^1$ \and Marin Costes$^2$}
\authorrunning{S. Apers \and M. Costes}
\institute{$^1$ Université Paris Cité, CNRS, IRIF, Paris, France\\
$^2$ Centre for Quantum Information and Communication, \'Ecole polytechnique de Bruxelles, Universit\'e libre de Bruxelles, Belgium\\
\email{marin.costes@ulb.be}}

\begin{document}
\maketitle

\begin{abstract}
Random walks are a standard tool for search problems in which a state can be updated locally and tested for being marked.
When updating the state and checking whether it is marked have different costs, two classical strategies optimize different parts of the cost: checking after every step is optimal in the number of updates, while repeatedly checking only after mixing is optimal in the number of checks.
For a single marked state \(m\) and a walk started from its stationary distribution \(\pi\), Dohotaru and H{\o}yer gave an algorithm optimal for both costs, which checks after blocks of a fixed length; their argument is sketched through quantum walks, and they observe that they know of no classical proof.
We give a short and self-contained classical proof of such a tradeoff, for arbitrary irreducible Markov chains.
The algorithm replaces the original transition matrix \(P\) by the averaged walk
\[
    \overline P_\tau = \frac{1}{\tau}\sum_{k=1}^{\tau} P^k,
\]
where \(\tau\) is of order \(\pi(m)\HT(m)\), with \(\pi(m)\) the stationary probability of \(m\) and \(\HT(m)\) the expected number of steps to reach \(m\) from stationarity.
Using a coupling with the original walk and Kac's lemma, we prove directly that the averaged walk hits the marked state in \(O(1/\pi(m))\) checks in expectation.
The resulting search cost is
\[
    S + O(\HT(m))U + O(1/\pi(m))C
\]
in expectation, where \(S\), \(U\), and \(C\) denote setup, update, and checking costs.
\end{abstract}

\section{Introduction}

A generic random walk search algorithm may be described as follows.
We have a finite state space \(X\), a Markov chain \(P\) on \(X\), and an unknown set \(M\subseteq X\) of marked states.
The algorithm starts from a sample of the stationary distribution, repeatedly updates the current state according to \(P\), and occasionally checks whether the current state is marked.
This abstraction is particularly useful when states carry data structures: moving to a neighboring state can reuse previous work, while checking whether the current state is marked may have a different cost.

The standard cost model separates three operations.
The setup cost \(S\) is the cost of sampling from the stationary distribution.
The update cost \(U\) is the cost of simulating one step of \(P\).
The checking cost \(C\) is the cost of testing whether the current state is marked.
This model is common in the quantum walk search literature \cite{MNRS11,Santha08}, where it is used to compare classical random walk search strategies and their quantum analogues.
It is already meaningful classically: if \(C\) is large, one wants to check rarely; if \(U\) is large, one wants to avoid unnecessary walk steps.

Two elementary strategies illustrate the tension.
Checking after every step finds a marked state in \(O(\HT(M))\) updates and checks, where \(\HT(M)\) is the hitting time from stationarity.
This uses an optimal number of updates, but can use many more checks than necessary.
On the other hand, if \(P\) is reversible with spectral gap \(\delta\), checking after blocks of length \(O(1/\delta)\) finds a marked state with \(O(1/\pi(M))\) checks (see Lemma~\ref{lem:check-optimal}), which is optimal as a bound in terms of \(\pi(M)\) alone (see Remark~\ref{rem:checks}), but can require \(O(1/(\delta\pi(M)))\) updates, which may be larger than \(\HT(M)\).

For a single marked state \(m\), this tradeoff is avoidable.
The same cost split appears in the work of Dohotaru and H{\o}yer on controlled quantum amplification \cite{DohotaruHoyer17}, as a byproduct of their quantum framework.
In \cite[Section 6]{DohotaruHoyer17} they consider the walk that checks after every \(\ell\) steps of \(P\), with \(\ell\) of the order of \(\pi(m)\HT(m)\), and sketch in one sentence a proof that the quantum hitting time of the Szegedy walk associated with \(P^\ell\) is constant.
Their argument requires \(P\) to be reversible, no classical proof is given, and they explicitly write that they know of no way to obtain the statement without quantum arguments.
Our goal is to provide such a proof, at an elementary level and for arbitrary irreducible chains.
A fixed block length need not suffice: on an even cycle, an even power of the simple random walk preserves parity and never reaches \(m\) from half of the stationary starts.
Instead of checking after a fixed number of steps, we draw an independent random integer \(B\) uniformly in \(\{1,\ldots,\tau\}\), perform \(B\) updates, and then check.
Equivalently, we run the averaged Markov chain
\[
    \overline P_\tau = \frac{1}{\tau}\sum_{k=1}^{\tau} P^k.
\]
For \(\tau \geq \pi(m)\HT(m)\), we prove that this averaged chain has hitting time \(O(1/\pi(m))\).
Thus the expected number of checks matches the strategy based on the spectral gap up to a constant, while the expected number of updates is \(O(\tau/\pi(m))=O(\HT(m))\), since \(\HT(m)=\Omega(1/\pi(m))\).

\paragraph{Contribution.}
The main technical contribution is the following direct theorem about an averaged Markov chain, stated informally here.

\begin{quote}
Let \(P\) be a finite irreducible Markov chain with stationary distribution \(\pi\), and let \(m\) be a single marked state.
If \(\tau \geq \pi(m)\HT(m)\), then the averaged chain \(\overline P_\tau\) reaches \(m\), from stationarity, after \(O(1/\pi(m))\) steps in expectation.
\end{quote}

The proof is short but somewhat delicate.
We couple the averaged chain to the original walk by viewing one step of \(\overline P_\tau\) as a randomly chosen block of \(P\)-steps.
The first part of the analysis counts the inspections made before the original walk reaches \(m\), which Wald's identity bounds by \(O(\HT(m)/\tau)\).
For the remaining part, the original walk has just hit \(m\), so Kac's lemma (Lemma \ref{lem:kac}) applied to the averaged chain guarantees that it returns to \(m\) after \(1/\pi(m)\) further inspections in expectation.

\paragraph{Quantum motivation.}
The question is motivated by the setup/update/check accounting used in quantum walk search \cite{MNRS11,Santha08}.
In that literature, the two elementary algorithms above are usually presented as intuition for quantum algorithms, and it is through such a quantum framework that Dohotaru and H{\o}yer arrived at the same classical cost split for a unique marked element.
Our contribution is to isolate the corresponding classical mechanism: a random block length, rather than a fixed one, gives a concise proof based on Kac's lemma.

\section{Preliminaries}

Let \(X\) be a finite set and let \(P\) be an irreducible Markov chain on \(X\).
We write \(P(y,x)\) for the probability to move from \(x\) to \(y\), the column-stochastic convention, and \((X_t)_{t\geq 0}\) for a trajectory of the chain.
The chain has a unique stationary distribution \(\pi\), satisfying \(P\pi=\pi\).
For \(A\subseteq X\), write \(\pi(A)=\sum_{x\in A}\pi(x)\).
We identify a state \(x \in X\) with the distribution concentrated on \(x\), so that \(P^j x\) denotes the law of \(X_j\) when \(X_0=x\).
For a distribution \(\mu\) on \(X\), we write \(\Prb_\mu\) and \(\E_\mu\) for the law and the expectation of a chain started from \(X_0\sim\mu\), the chain in question being clear from the context.

For \(A\subseteq X\), let
\[
    T_A = \min\{t\geq 0 : X_t\in A\}
\]
be the hitting time of \(A\), and let \(T_m=T_{\{m\}}\) for a single state \(m\).
Unless otherwise specified, hitting times are taken when \(X_0\sim \pi\), and we write
\[
    \HT(A) = \E_\pi[T_A].
\]
For a Markov chain \(Q\) having \(\pi\) as stationary distribution, we write \(\HT_Q(m)\) for the same quantity when the steps are taken according to \(Q\) instead of \(P\).
We use it for a power of \(P\) in Lemma~\ref{lem:check-optimal}, and for the averaged chain in Theorem~\ref{thm:main}.

We use the following standard form of Kac's lemma.

\begin{lemma}[Kac's lemma~\cite{Kac47,LPW17}]
\label{lem:kac}
Let \(Q\) be a finite irreducible Markov chain on \(X\) with stationary distribution \(\pi\), and let \((W_t)_{t\geq 0}\) be a trajectory of \(Q\).
For every state \(m\),
\[
    \E_m[\min\{t>0 : W_t=m\}] = \frac{1}{\pi(m)}.
\]
\end{lemma}

\section{The Search Model}

An instance consists of a finite irreducible Markov chain \(P\) with stationary distribution \(\pi\), together with a marked state \(m\).
The algorithm has access to three operations:
\begin{itemize}
    \item \emph{Setup}: sample \(X_0\sim \pi\), at cost \(S\).
    \item \emph{Update}: from the current state \(x\), sample one transition according to \(P(\cdot,x)\), at cost \(U\).
    \item \emph{Check}: decide whether the current state is \(m\), at cost \(C\).
\end{itemize}
The goal is to output \(m\) with constant probability.
% We focus on the single marked state case.

The two natural costs in this model are the expected number of updates needed to reach \(m\) along \(P\), namely \(\HT(m)\), and the expected number of independent stationary samples needed to see \(m\), namely \(1/\pi(m)\).
The next section records two algorithms that match these bounds separately; the averaged-walk algorithm then matches them simultaneously.

\section{Two Elementary Search Algorithms}
\label{sec:elementary}

Two elementary algorithms optimize opposite parts of the cost.
The first checks after every update.
It uses \(O(\HT(m))\) updates and \(O(\HT(m))\) checks, and therefore has cost
\[
    S + O(\HT(m))(U+C).
\]
This is optimal in the number of updates: the algorithm finds \(m\) at the first visit, and any search that follows \(P\) cannot output \(m\) before that time.
It can be wasteful when \(C\) is large.

The second checks only after long blocks.
If \(P\) is reversible, write \(\lambda_*\) for the largest absolute value of an eigenvalue of \(P\) other than \(1\), and \(\delta=1-\lambda_*\) for the absolute spectral gap.
A block of length \(O(1/\delta)\) is then enough to decorrelate from the previously checked state.
This gives \(O(1/\pi(m))\) checks and cost
\[
    S + O\!\left(\frac{1}{\pi(m)}\right)C
      + O\!\left(\frac{1}{\delta\pi(m)}\right)U.
\]
Unlike the first algorithm, the constant success probability is not immediate; it is stated in Lemma~\ref{lem:check-optimal}.
Optimality of the number of checks, as a bound in terms of \(\pi(m)\) alone, is Remark~\ref{rem:checks}.
The update cost can be larger than \(\HT(m)\).

The averaged-walk algorithm that we propose in the next section combines the two guarantees for a single marked state.

\begin{lemma}
\label{lem:check-optimal}
Let \(P\) be a finite irreducible reversible Markov chain with absolute spectral gap \(\delta>0\), and let \(m\) be a marked state.
From a stationary start, the algorithm that repeatedly checks the current state and then applies \(\lceil 1/\delta\rceil\) steps of~\(P\), for \(R=\Theta(1/\pi(m))\) iterations, outputs \(m\) with probability at least \(1/2\).
\end{lemma}

\begin{proof}
Note first that \(\delta>0\) forces \(P\) to be aperiodic, since a periodic chain has an eigenvalue of modulus \(1\) other than \(1\).
Consider the walk \(Q=P^{\lceil 1/\delta\rceil}\).
Its absolute spectral gap is
\[
    \delta' = 1-\lambda_*^{\lceil 1/\delta\rceil}
    \geq 1-(1-\delta)^{1/\delta}
    \geq 1-e^{-1},
\]
hence bounded below by a positive constant.
Since \(Q\) is again reversible with respect to \(\pi\), the standard spectral bound on hitting times of reversible chains yields \(\HT_Q(m)=O(1/(\delta'\pi(m)))=O(1/\pi(m))\) \cite[Proposition 1]{MNRS11}.
By Markov's inequality, \(R=\lceil 2\HT_Q(m)\rceil\) checks suffice for success probability at least \(1/2\). $\qed$
\end{proof}

\begin{remark}
\label{rem:checks}
No bound \(o(1/\pi(m))\) on the number of checks can hold for every walk.
Let \(X\) be the complete graph on \(n\) vertices, with uniform transitions \(P(y,x)=1/n\) for all \(x,y\), and let \(m\) be a single marked vertex, so that \(\pi(m)=1/n\).
From every state, one step of \(P\) is already an independent sample from \(\pi\).
Thus updates do not create any search advantage over testing stationary samples, and each check of a freshly sampled state succeeds with probability \(1/n\), independently of previous unsuccessful checks.
The number of checks is therefore geometric of parameter \(1/n\), and \(\Omega(1/\pi(m))\) checks are necessary in expectation.
\end{remark}

Reversibility is used at one single place, namely the bound \(\HT_Q(m)=O(1/\pi(m))\) for a chain with constant spectral gap, which rests on the spectral decomposition of a transition matrix that is self-adjoint in \(\ell^2(\pi)\).
This is the classical estimate used by Magniez, Nayak, Roland and Santha for reversible chains \cite{MNRS11}.
For a chain that is not reversible, \(P\) need not be normal, and the modulus of its eigenvalues no longer controls hitting times; the same algorithm can still be used with blocks long enough for the current distribution to be within total variation distance \(\pi(m)/2\) of \(\pi\), which again gives \(O(1/\pi(m))\) checks, at the price of a logarithmic factor in the length of the blocks.
Nothing of this is needed in the next section: our Theorem~\ref{thm:main} holds for every finite irreducible chain.

\section{Averaged-Walk Search}

Fix an integer \(\tau\geq 1\), and define
\[
    \overline P_\tau = \frac{1}{\tau}\sum_{k=1}^{\tau}P^k.
\]
One step of \(\overline P_\tau\) is implemented by drawing \(B\) uniformly from \(\{1,\ldots,\tau\}\) and performing \(B\) consecutive updates of \(P\).
Since \(\pi\) is stationary for \(P\), it is also stationary for \(\overline P_\tau\). We now introduce our search algorithm using this averaged walk, which simply applies \(\overline P_\tau\) and checks the state after each step.

\begin{algorithm}
\caption{Averaged-walk search}
\label{alg:averaged-search}
\begin{algorithmic}[1]
\State \textbf{input:} irreducible Markov chain \(P\) with stationary distribution \(\pi\), marked state \(m\), block length \(\tau\) and number of checks \(R\)
\State \textbf{setup:} sample \(X_0\sim \pi\)
\For{\(r=0,1,\ldots,R-1\)}
    \State \textbf{check} whether the current state is \(m\), and stop if so
    \State draw \(B_r\) uniformly from \(\{1,\ldots,\tau\}\)
    \State \textbf{update} by applying \(B_r\) times the transition matrix \(P\)
\EndFor
\end{algorithmic}
\end{algorithm}

We now prove the main theorem.

\begin{theorem}
\label{thm:main}
Let \(P\) be a finite irreducible Markov chain with stationary distribution \(\pi\), and let \(m\in X\).
Let \(\tau\) be an integer satisfying
\[
    \tau \geq \pi(m)\HT(m).
\]
Then the averaged chain \(\overline P_\tau\), started from \(\pi\), satisfies
\[
    \HT_{\overline P_\tau}(m)
    \leq \frac{2\HT(m)}{\tau+1} + 2 + \frac{e}{\pi(m)}
    \leq \frac{e+4}{\pi(m)}.
\]
Consequently, when fixing \(\tau=\lceil \pi(m)\HT(m)\rceil\) and \(R=\lceil \frac{2(e+4)}{\pi(m)}\rceil\), Algorithm~\ref{alg:averaged-search} finds \(m\) with constant probability using \(O(1/\pi(m))\) checks and \(O(\HT(m))\) updates.% in expectation.
\end{theorem}

For the proof we have to couple the original chain with the averaged chain. Let \(B_1,B_2,\ldots\) be independent uniform random variables on \(\{1,\ldots,\tau\}\), which we call the block lengths, and let
\[
    \sigma_r = \sum_{k=1}^r B_k, \qquad \sigma_0=0.
\]
We couple the original chain \((X_t)_{t\geq 0}\), started from \(X_0\sim\pi\), with the averaged chain by setting
\[
    Z_r = X_{\sigma_r}.
\]
Then \((Z_r)_{r\geq 0}\) is exactly the chain \(\overline P_\tau\).
Let
\[
    Y_t = \begin{cases} 1 & \text{if } \exists r\geq 0 : \sigma_r=t \\ 0 & \text{otherwise}
        \end{cases}
\]
indicate whether time \(t\) of the original walk is inspected by the averaged walk; we call the times \(\sigma_r\) the inspected times.

\begin{figure}[t]
\centering
\input{figures/proof_sketch}
\caption{The coupling in the proof of Theorem~\ref{thm:main}. Blue arcs are single steps of \(P\), and red arcs are steps of \(\overline P_\tau\), of random lengths \(B_r\). The original walk first visits \(m\) at time \(T_m\), between two inspected times. The first \(r_0\) steps of the averaged chain end before \(T_m\) and cannot see \(m\); after \(N'\) further steps, the averaged chain reaches \(m\), at index \(N=r_0+N'\).}
\label{fig:proof}
\end{figure}

\paragraph{Proof idea.}
Figure~\ref{fig:proof} illustrates the proof, which splits the steps of the averaged chain into two phases.
In the first phase, the original walk has not yet reached \(m\), so no inspection can succeed.
The original walk needs \(\HT(m)\) steps of \(P\) on average to reach \(m\), and a step of \(\overline P_\tau\) contains \((\tau+1)/2\) steps of \(P\) on average.
The first phase therefore lasts about \(2\HT(m)/\tau\) steps of \(\overline P_\tau\), which is \(O(1/\pi(m))\) when \(\tau\geq\pi(m)\HT(m)\).
Wald's identity (Lemma~\ref{lem:wald}) makes this count rigorous.
In the second phase, the original walk has just visited \(m\), shortly before the inspection of index \(r_0\).
The averaged chain is then almost in the situation of having just left \(m\), and Kac's lemma applied to \(\overline P_\tau\), which has the same stationary distribution \(\pi\), says that it returns to \(m\) after \(1/\pi(m)\) steps in expectation.
The word ``almost'' is quantified by Lemma~\ref{lem:checkpoints}: the delay \(\Delta_{T_m}\) between the visit at \(T_m\) and the next inspection takes each value with probability at most \(e/\tau\), within a constant factor of the probability \(1/\tau\) of each length of one step of \(\overline P_\tau\) started from \(m\).

\begin{lemma}
\label{lem:checkpoints}
For every \(t\geq 1\),
\[
    \Prb(Y_t=1) \leq \frac{e}{\tau}.
\]
\end{lemma}

\begin{proof}
We first prove the bound for \(1\leq t\leq \tau\).
In this range, any subset of \(\{1,\ldots,t-1\}\) can be the set of inspected times before \(t\), because all block lengths are at most \(\tau\).
Time \(0\) is always inspected, since \(\sigma_0=0\).
If there are \(p\) inspected times in \(\{1,\ldots,t-1\}\), the \(p+1\) blocks from time \(0\) through time \(t\) are fixed, and each has probability \(1/\tau\).
Therefore
\[
    \Prb(Y_t=1)
    = \sum_{p=0}^{t-1} \binom{t-1}{p}\tau^{-(p+1)}
    = \frac{1}{\tau}\left(1+\frac{1}{\tau}\right)^{t-1}
    \leq \frac{e}{\tau}.
\]

Now assume \(t>\tau\), and suppose by induction that the bound holds for all smaller times.
Let
\[
    L_t = \max\{s<t : Y_s=1\}
\]
be the last inspected time before \(t\). %, with the convention \(L_t=0\) if there is no such time.
If \(Y_t=1\), then necessarily \(t-L_t\in\{1,\ldots,\tau\}\), and the next block after \(L_t\) must have length \(t-L_t\).
Thus, decomposing according to the distance \(d=t-L_t\),
\[
\begin{aligned}
    \Prb(Y_t=1)
    &= \sum_{d=1}^{\tau}
        \Prb(Y_t=1 \ \text{and}\ L_t=t-d) \\
    %&= \sum_{d=1}^{\tau}
        %\Prb(Y_{t-d}=1)\Prb(\text{the next jump after }t-d\text{ has length }d) \\
    &= \frac1{\tau}\sum_{d=1}^{\tau}\Prb(Y_{t-d}=1).
\end{aligned}
\]
By the induction hypothesis,
\[
    \Prb(Y_t=1)
    \leq \frac1{\tau}\sum_{d=1}^{\tau}\frac{e}{\tau}
    = \frac{e}{\tau}. \qquad\qquad\qquad \qed
\]
\end{proof}

We will also use the following form of Wald's identity \cite[Theorem~2.6.2]{Durrett19}.

\begin{lemma}[Wald's identity] \label{lem:wald}
Let \((X_n)_{n\geq 1}\) be i.i.d.\ real random variables with finite expectation, and let \(N\) be a stopping time for this sequence: for every \(n\geq 0\), the event \(\{N=n\}\) depends only on \(X_1,\ldots,X_n\).
If \(\E[N]<\infty\), then
\[
    \E[X_1+\cdots+X_N] = \E[N]\,\E[X_1]
\]
%where the sum on the left is \(0\) if \(N=0\).
\end{lemma}

We now turn to the proof of our main theorem.

\begin{proof}[Theorem~\ref{thm:main}]
Recall that \(T_m=\min\{t\geq 0 : X_t=m\}\), and let
\[
    N = \min\{r\geq 0 : Z_r=m\}
\]
be the corresponding hitting time among the inspected times, so that \(\E[N]=\HT_{\overline P_\tau}(m)\).
Let \(r_0=\min\{r\geq 0: \sigma_r\geq T_m\}\) be the index of the first inspection occurring at or after the first visit of the original walk to \(m\), and let \(N'=\min\{s\geq 0 : Z_{r_0+s}=m\}\).
Every inspection of index \(r<r_0\) occurs at a time \(\sigma_r<T_m\), at which the original walk is not at \(m\).
Hence \(N\geq r_0\), and
\[
    \E[N] = \E[r_0] + \E[N'].
\]

We first bound \(\E[r_0]\).
The block lengths are independent of the walk, so we may condition on \(T_m=i\).
Then \(r_0\) is a stopping time for the sequence \((B_r)_{r\geq 1}\).
Every block has length at least \(1\), so \(r_0\leq i\) and \(\E[r_0\mid T_m=i]\leq i<\infty\).
Lemma~\ref{lem:wald} applies to this conditional law, and since \(\E[B_1]=(\tau+1)/2\),
\[
    \frac{\tau+1}{2}\,\E[r_0\mid T_m=i]
    = \E[\sigma_{r_0}\mid T_m=i]
    \leq i+\tau-1.
\]
If \(r_0=0\), then \(\sigma_{r_0}=0\).
If \(r_0\geq 1\), then \(\sigma_{r_0-1}<i\) and \(B_{r_0}\leq\tau\), so \(\sigma_{r_0}\leq i+\tau-1\).
Averaging over \(i\) and using \(\E[T_m]=\HT(m)\),
\[
    \E[r_0]
    \leq \frac{2\bigl(\HT(m)+\tau-1\bigr)}{\tau+1}
    \leq \frac{2\HT(m)}{\tau+1} + 2.
\]

It remains to bound \(\E[N']\).
For a time \(i\), let
\[
    \Delta_i = \min\{j\geq 0 : Y_{i+j}=1\}
\]
be the delay until the next inspected time.
Since all block lengths are at most \(\tau\), every interval of \(\tau\) consecutive times contains an inspected time, and hence \(\Delta_i\leq \tau-1\).
Condition on the event \(T_m=i\) and \(\Delta_i=j\).
Since there is a unique marked element, the original walk is at \(m\) at time \(i\), and has not visited \(m\) before.
The next inspection occurs at time \(i+j\), so the inspected state is \(X_{i+j}\).
By the Markov property, this state has law \(P^j m\).
The inspected times are independent of \((X_t)\), and conditionally on \(\Delta_i=j\) the block lengths after time \(i+j\) are fresh uniform draws, so from that inspection onward the averaged chain is distributed as \(\overline P_\tau\) started from \(P^j m\), and therefore
\[
    \E[N' \mid T_m=i,\ \Delta_i=j] = \E_{P^j m}[N].
\]
Using also the independence of the inspected times and the Markov chain, and the fact that all terms are nonnegative, we get
\[
    \E[N']
    =
    \sum_{i\geq 0}\Prb(T_m=i)
    \sum_{j=0}^{\tau-1}\Prb(\Delta_i=j)\,\E_{P^j m}[N].
\]
The term \(j=0\) vanishes since \(\E_m[N]=0\), and for \(j\geq 1\) the event \(\Delta_i=j\) implies \(Y_{i+j}=1\), so that Lemma~\ref{lem:checkpoints} applies and gives
\[
    \E[N']
    \leq
    \frac{e}{\tau}
    \sum_{j=1}^{\tau-1} \E_{P^j m}[N]
    \leq
    \frac{e}{\tau}
    \sum_{j=1}^{\tau} \E_{P^j m}[N].
\]
The last average is controlled by Kac's lemma applied to \(\overline P_\tau\).
This is legitimate: \(\pi\) is stationary for \(\overline P_\tau\), and \(\overline P_\tau\geq \frac1\tau P\) entrywise, so \(\overline P_\tau\) is irreducible as well.
After one step from \(m\), the averaged chain reaches the distribution \(\overline P_\tau m=\frac1\tau\sum_{j=1}^{\tau}P^j m\), so
\[
    \E_m[\min\{r>0: Z_r=m\}]
    =
    1 + \frac1\tau\sum_{j=1}^{\tau}\E_{P^j m}[N].
\]
By Lemma~\ref{lem:kac}, the left-hand side equals \(1/\pi(m)\).
Thus
\[
    \frac1\tau\sum_{j=1}^{\tau}\E_{P^j m}[N]
    \leq \frac1{\pi(m)},
\]
and so \(\E[N']\leq e/\pi(m)\).
Combining the two bounds proves the first inequality.
The second uses \(\frac{2\HT(m)}{\tau+1}\leq \frac{2}{\pi(m)}\), from \(\tau\geq \pi(m)\HT(m)\), and \(2\leq \frac{2}{\pi(m)}\), from \(\pi(m)\leq 1\). $\qed$
\end{proof}

\begin{corollary}
\label{cor:cost}
Choose \(\tau\geq \pi(m)\HT(m)\) and run Algorithm~\ref{alg:averaged-search} with
\[
    R \geq \frac{2(e+4)}{\pi(m)}.
\]
Then the algorithm outputs \(m\) with probability at least \(1/2\).
\end{corollary}

\begin{proof}
Let \(N\) be the hitting time of \(m\) for the averaged chain, as in the proof of Theorem~\ref{thm:main}.
By Markov's inequality and Theorem~\ref{thm:main},
\[
    \Prb(N\geq R)
    \leq \frac{\E[N]}{R}
    \leq \frac{(e+4)/\pi(m)}{R}
    \leq \frac12. \qquad\qquad \qed
\]
\end{proof}

\section{Cost and Optimality}

Choose \(\tau=\lceil \pi(m)\HT(m)\rceil\).
By Corollary~\ref{cor:cost}, \(R=\Theta(1/\pi(m))\) iterations of Algorithm~\ref{alg:averaged-search} are enough for constant success probability.
Each iteration performs one check and an expected \((\tau+1)/2\) updates.
Since the algorithm may stop early, \(R(\tau+1)/2\) is an upper bound on the expected number of updates.
Since \(\tau\leq \pi(m)\HT(m)+1\), and \(\Prb(T_m\leq t)\leq (t+1)\pi(m)\) yields \(\HT(m)=\Omega(1/\pi(m))\) for \(\pi(m)\leq 1/2\), the expected cost is therefore
\[
    S
    + O\!\left(\frac{1}{\pi(m)}\right)C
    + O\!\left(\frac{\tau}{\pi(m)}\right)U
    =
    S
    + O\!\left(\frac{1}{\pi(m)}\right)C
    + O(\HT(m))U.
\]

This matches the two elementary algorithms at once: \(O(\HT(m))\) updates, as when checking after every step, and \(O(1/\pi(m))\) checks, as when checking after mixing (Lemma~\ref{lem:check-optimal}).
The theorem should therefore be read as matching both costs at once, for a fixed chain \(P\).

\paragraph{Example: the cycle.}
Consider the lazy random walk on the cycle with \(n\) vertices and one marked vertex \(m\).
The stationary distribution is uniform, so \(\pi(m)=1/n\).
The hitting time from stationarity is \(\Theta(n^2)\), while the spectral gap is \(\Theta(1/n^2)\).
Table~\ref{tab:cycle} compares the two elementary algorithms with the averaged walk.
Checking after every step is optimal in the number of updates, but uses \(\Theta(n^2)\) checks.
Checking after mixing is optimal in the number of checks, but uses \(\Theta(n^3)\) updates, a factor \(n\) above the hitting time.
The averaged walk chooses \(\tau=\Theta(n)\) and matches both: \(\Theta(n)\) checks and \(\Theta(n^2)\) updates.

\begin{table}[ht]
\centering
\caption{Expected number of checks and updates on the cycle, up to constant factors, for constant success probability.}
\label{tab:cycle}
\begin{tabular}{lcc}
\hline
Algorithm & Checks & Updates \\
\hline
Check after every step & \(n^2\) & \(n^2\) \\
Check after mixing & \(n\) & \(n^3\) \\
Averaged walk & \(n\) & \(n^2\) \\
\hline
\end{tabular}
\end{table}

\section{Limitations and Related Work}

The single marked state assumption is essential for the proof.
After the original walk first hits \(m\), the state is exactly known, and the remaining expectation can be related to the return time to \(m\).
For a marked set \(M\) with several states, the distribution of the first hit inside \(M\) can be biased in a way that is not captured only by \(\pi(M)\), and the Kac argument no longer applies directly.
Finding a useful multi-marked analogue, or a counterexample showing that none exists at this level of generality, is a natural open question.

The setup/update/check cost model comes from quantum walk search, especially the framework of Magniez, Nayak, Roland, and Santha \cite{MNRS11}.
Their work compares classical random walk search procedures with quantum analogues and obtains quantum speedups in terms of the stationary marked mass and spectral gap.
Later quantum walk frameworks refine the dependence on hitting time and checking cost, including controlled amplification and unified search frameworks \cite{DohotaruHoyer17,AGJ21}.
In particular, Dohotaru and H{\o}yer state a classical walk for a unique marked element using \(O(\HT(m))\) updates and \(O(1/\pi(m))\) checks, obtained by checking every \(\ell\) steps for a fixed \(\ell\) of the order of \(\pi(m)\HT(m)\).
Their justification consists of a proof sketch asserting that the quantum hitting time of the Szegedy walk of \(P^\ell\) is constant, so it relies on reversibility of \(P\) and on the correspondence between classical and quantum hitting times for a unique marked element.
It also depends on \(P^\ell\) being well behaved, which a fixed power need not be: for the simple random walk on an even cycle and an even \(\ell\), the chain \(P^\ell\) preserves the parity of the vertex, so from half of the stationary starts it never reaches \(m\).
Randomizing the block length removes both restrictions.
Our contribution is thus a classical proof of the tradeoff, valid for every irreducible chain: the averaged chain has hitting time \(O(1/\pi(m))\) as soon as \(\tau\geq \pi(m)\HT(m)\).

Kac's lemma is a classical recurrence theorem for Markov chains and ergodic systems \cite{Kac47,LPW17}.
Our use of it is elementary but central: the averaged chain is constructed so that missed visits to the marked state can be charged to a return time of the same chain.
Averaging consecutive powers of a Markov chain is also closely related to standard techniques for removing periodicity and defining averaged mixing times, although our objective is hitting a particular state with sparse inspections rather than convergence to stationarity.

\section{Conclusion}

We gave a simple Markov-chain proof of the classical random walk search tradeoff for the case of one marked state and separate update and checking costs.
The algorithm checks after a uniformly random number of walk steps rather than after a fixed block length.
This changes the inspected subsequence enough to guarantee \(O(1/\pi(m))\) checks, while preserving \(O(\HT(m))\) updates.

The main open directions are to understand the right analogue for multiple marked states, remove or reduce the need to know \(\pi(m)\HT(m)\), and determine whether the theorem is useful as a black-box classical component in quantum walk search constructions.

\section*{Acknowledgements}
This publication was made possible through the support of the ID 63683 from the John Templeton Foundation, as part of the “WithOut SpaceTime”
Project (WOST). The opinions expressed in this publication are those of the authors and do not necessarily reflect the views of the John Templeton Foundation. This work was also supported by the F.R.S.-FNRS under project CHEQS within the Excellence of Science (EOS) program.

\bibliographystyle{splncs04}
\bibliography{biblio}

\end{document}

%% file: macro_Tikz.tex
\newcommand{\bigPoint}[2]{%(name,position)
    \draw (#2,0.15) node[above]{\color{red} #1} ;
    \draw (#2,0) node{\color{red} $|$} ;
    \draw (#2,-0.15) node[below]{\color{red} $X_{#2}$} ;
}

\newcommand{\smallPoint}[1]{%(position)
    \draw (#1,0) node{\small $|$} ;
    \draw (#1,-0.15) node[below]{$X_{#1}$} ;
}

\newcommand{\bigJump}[3]{%(name,start,end)
    \draw[->,>=stealth,very thick,red] (#2,0) to[bend left] node[pos= 0.5,above]{#1} (#3,0);
}

\newcommand{\autoBigJump}[4]{%(name_s,name_e,start,end)
    \bigPoint{#1}{#3};
    \bigPoint{#2}{#4};
    \foreach \i in {#3,...,#4} \ifthenelse{\NOT \i=#3 \AND \NOT \i=#4}{ \smallPoint{\i}}{};
    \bigJump{$F$}{#3}{#4};
}

%% file: figures/proof_sketch.tex
\begin{tikzpicture}[x=0.52cm,y=0.55cm]
    \draw[->,thick] (0,0) -- (21,0) node[right]{$t$};
    \foreach \t in {0,...,19} {\draw[->,>=stealth,blue] (\t,0) to[out=-50,in=230] (\t+1,0);}

    \bigJump{}{0}{3};
    \bigJump{}{3}{6};
    \bigJump{$B_{r_0}$}{6}{10};
    \bigJump{}{10}{13};
    \bigJump{}{13}{17};
    \bigJump{}{17}{20};

    \foreach \t in {0,3,6,10,13,17,20} {\draw[red] (\t,0) node{$|$};}
    \node[red,above] at (0,0.3) {$Z_0$};
    \node[red,above] at (3,0.3) {$Z_1$};
    \node[red,above] at (6,0.3) {$Z_2$};
    \node[red,above] at (10,0.3) {$Z_{r_0}$};
    \node[red,above] at (20,0.3) {$Z_N=m$};

    \fill (8,0) circle (3pt);
    \fill[red] (20,0) circle (3pt);
    \node[below] at (0,-0.5) {$X_0\sim\pi$};
    \node[below] at (8,-0.5) {$X_{T_m}=m$};

    \draw[dashed,gray] (10,0) -- (10,-3);

    \draw[<->,>=stealth,thick] (8,-2) -- node[below]{$\Delta_{T_m}$} (10,-2);

    \draw[decorate,decoration={brace,mirror,amplitude=5pt}] (0,-3.4) -- node[below=6pt]{$r_0$ steps of $\overline P_\tau$} (10,-3.4);
    \draw[decorate,decoration={brace,mirror,amplitude=5pt}] (10,-3.4) -- node[below=6pt]{$N'$ steps of $\overline P_\tau$} (20,-3.4);
\end{tikzpicture}